\documentclass{article}

\usepackage{amsthm, amssymb, amsmath}
\usepackage{verbatim}
\usepackage{graphics}
\usepackage{graphicx}
\usepackage{setspace}
\usepackage{url}

\newtheorem{thm}{Theorem}
\newtheorem{prop}[thm]{Proposition}

\newtheorem{lem}[thm]{Lemma}

\theoremstyle{plain}
\newtheorem*{defn}{Definition}
\newtheorem*{nthm}{Theorem}

\newcommand{\R}{\mathbb{R}}

\newcommand{\pr}{\textrm{Pr}}
\newcommand{\sgn}{\textrm{sgn}}
\newcommand{\E}{\textrm{E}}
\newcommand{\dotp}[2]{\left\langle #1,#2\right\rangle}

\title{$k$-Independent Gaussians Fool Polynomial Threshold Functions}
\author{Daniel M. Kane}

\begin{document}

\maketitle

\section{Introduction}

In this paper we consider the ability of limited independence to fool polynomial threshold functions (PTFs).  We recall that a (degree-$d$) polynomial threshold function is a function of the form $f(x) = \sgn(p(x))$ for some $n$-dimensional polynomial $p$ of degree at most $d$.  There has been recent interest in polynomial threshold functions in several areas of computer science.  This paper expands on previous work in derandomizing polynomial threshold functions using limited independence.

We say that a random variables $X$ fools a family of functions with respect to some distribution $Y$ if for every function, $f$, in the family
$$
|\E[f(X)]-\E[f(Y)]| = O(\epsilon).
$$
In this paper we will be interested in the case where the family is of all degree-$d$ polynomial threshold functions in $n$-variables, and $Y$ is either an $n$-dimension Gaussian distribution, and in particular the case where $X$ is an arbitrary family of $k$-independent Gaussian random variables.  In particular, we prove that

\begin{thm}\label{mainTheorem}
Let $d>0$ be an integer and $\epsilon>0$ a real number, then there exists a $k=O_d\left( \epsilon^{-2^{O(d)}}\right)$, so that for any degree $d$ polynomial $p$ and any $k$-independent family of Gaussians $X$ and fully independent family of Gaussians $Y$
$$
\left| \E[\sgn(p(X))]-\E[\sgn(p(Y))]\right| = O(\epsilon).
$$
\end{thm}

There has been a significant amount of recent work on the problem of fooling low degree polynomial threshold functions of Gaussian or Bernoulli random variables, especially via limited independence.  It was shown in \cite{DGJSV} that $\tilde O(\epsilon^{-2})$-independence is sufficient to fool degree-1 polynomial threshold functions of Bernoulli random variables, and show that this is tight up to polylogarithmic factors.  In \cite{DKN} it was shown that $\tilde O(\epsilon^{-9})$-independence sufficed for degree-2 polynomial threshold functions of Bernoullis and that $O(\epsilon^{-2})$ and $O(\epsilon^{-8})$ suffices for degree 1 and 2 polynomial threshold functions of Gaussians.  The degree 1 case was also extended by \cite{BLY}, who show that limited independence fools threshold functions of polynomials that can be written in terms of a small number of linear polynomials. Finally, in \cite{MZ} a more complicated pseudorandom generator for degree-$d$ polynomial threshold functions of Bernoulli variables is developed with seed length $2^{O(d)}\log(n) \epsilon^{-8d-3}$.  As far as we are aware, our paper is the first result to show that degree-$d$ polynomial threshold functions are fooled by $k$-independence for any $k$ depending only on $\epsilon$ and $d$ for any $d\geq 3$.

\section{Overview}

We prove Theorem \ref{mainTheorem} first by proving our result for multilinear polynomials, and then finding a reduction to the general case.  In particular we prove
\begin{prop}\label{mainProp}
Let $d>0$ be an integer and $\epsilon>0$ a real number, then there exists a $k=O_d\left( \epsilon^{-2^{O(d)}}\right)$, so that for any degree $d$ multilinear polynomial $p:\R^n\rightarrow \R$ and any $k$-independent family of Gaussians $X$ and fully independent family of Gaussians $Y$
$$
\left| \E[\sgn(p(X))]-\E[\sgn(p(Y))]\right| = O(\epsilon).
$$
\end{prop}

We define the notation $A\approx_\epsilon B$ to mean $|A-B|=O(\epsilon)$.

The proof of Proposition \ref{mainProp} will be analogous to the proof of the main Theorem in \cite{DKN}.  Our basic idea is as follows.

In Section \ref{momentSec} we prove bounds on the moments of multilinear Gaussian polynomials.  These results are essentially a reworking of the main result of \cite{Moments}.

In Section \ref{structureSec}, we use these bounds to prove a structure Theorem for multilinear polynomials.  In particular, we prove that we can write $p(X)$ in the form $h(P_1(X),P_2(X),\ldots,P_N(X))$ where $h$ is a polynomial and $P_i(X)$ are multilinear polynomials with relatively small higher moments.  More specifically, the polynomials $P_i$ will be split into $d$ different classes, with the $i^{th}$ class consisting of $n_i$ polynomials each of whose $m_i^{th}$ moments are $O_d(m_i)^{m_i/2}$.  This decomposition allows us to write $f(X)=\sgn(P(X))$ as $\sgn(h(P_1(X),\ldots,P_N(X)))$.

From here we make use of the FT-Mollification method (see \cite{DKN} for another example of this technique).  The basic idea will be to approximate $\sgn\circ h$ by some smooth function $\tilde h$, and let $\tilde f(X) = \tilde h(P_1(X),\ldots,P_N(X))$, which we do in Section \ref{FTMSec}.  Our general strategy now will be to prove the sequence of approximations:
$$
\E[f(Y)] \approx_\epsilon \E[\tilde f(Y)] \approx_\epsilon \E[\tilde f(X)] \approx_\epsilon \E[f(X)].
$$
The middle equality will be proved by approximation $\tilde f$ by one of it's Taylor polynomials.  This is a polynomial, and hence its expectation is preserved under limited independence.  The Taylor error can again be bounded by a polynomial, which will have small expectation since the $P_i$ have small moments.  We cover this in Section \ref{TaylorSec}.

The first approximation above holds roughly because $\tilde f$ approximates $f$ everywhere except near places where $f$ changes sign.  The result will hold due to anti-concentration results for $p(Y)$.  The last approximation similarly holds because of anti-concentration of $p(X)$.  Although anticoncentration of the $k$-independent $X$ can be proven using the above techniques applied to some other function $g$ for which $\tilde g$ is an upper bound for $f$, we deal with the problem indirectly.  In particular, we show that $\E[f(X)]$ can be bounded on either side by $\E[\sgn(p(Y)+c)]+O(\epsilon)$ for $c$ a small constant, and use anticoncentration of $p(Y)$.  We cover this in Section \ref{AntiConSeq}.

Our application of FT-Mollification is complicated by the fact that our moment bounds on the $P_j$ are not uniform in $j$.  To deal with this, we will construct $\tilde h$ to have different degrees of smoothness in different directions, and the parameter $C_i$ will describe the amount of smoothness along the $i^{th}$ set of coordinates (corresponding the the $i^{th}$ class of the $P_j$).  This forces us to come up with modified techniques for producing $\tilde h$ and dealing with the Taylor polynomial and Taylor error.

In Section \ref{ReductionSec}, we reduce the general case to the case of multilinear polynomials by approximating $p(X)$ by a multilinear polynomial in some larger number of variables.

Finally, in Section \ref{kIndepConc}, we discuss the actual requirements for $k$ and the possibility of extended our results to the Bernoulli setting.

\section{Moment Bounds}\label{momentSec}

In this Section, we prove a bound on the moments of arbitrary degree-$d$ multilinear polynomials of Gaussians.  Our bound is based on the main result of \cite{Moments}.  It should be noted that this result is the only reason that we restrict ourselves for most of this paper to the case of multilinear polynomials, as it will make our bound easier to state and work with.

Throughout this Section, we will refer to two slightly different notions that of a multilinear polynomial and that of a multilinear form.  For our purposes, a multilinear polynomial $p(X)$ ($X$ has $n$ coordinates) will be a polynomial so that the degree of $p$ with respect to any of the coordinates of $X$ is at most 1. A multilinear form will be a polynomial $q(X^1,X^2,\ldots,X^m)$ (here each of the $X^i$ may themselves have several coordinates) so that $q$ is linear (homogeneous degree 1) in each of the $X^i$.  We call such a $q$ symmetric if it is symmetric with respect to interchanging the $X^i$.  Finally, we note that to every homogeneous multilinear polynomial $p$ of degree $d$, there is an associated multilinear form $q(X^1,\ldots,X^d)$, which is the unique symmetric multilinear form so that $p(X)=q(X,\ldots,X).$

Before we can state our results we need a few more definitions.

\begin{defn}
Let $p:\R^n\rightarrow \R$ be a homogeneous degree-$d$ multilinear polynomial.  Let $X_i, 1\leq i\leq n$ be independent standard Gaussians.  For a integers $1\leq \ell \leq d$ define $M_\ell(p)$ in the following way.  Consider all possible choices of: a partition of $\{1,\ldots,n\}$ into sets $S_1,S_2,\ldots,S_\ell$; a sequence of integers $d_i \geq 1, 1\leq i\leq \ell$ so that $d=\sum_{i=1}^\ell d_i$; a sequence of multilinear polynomials $p_i, 1\leq i \leq \ell$ so that $p_i$ depends only on the coordinates in $S_i$, $p_i$ is homogeneous of degree $d_i$, and $\E[p_i(X)^2]=1$.  We let $M_\ell(p)$ be the supremum over all choices of $S_i,d_i,p_i$ as above of $$\left(\E\left[p(X)\prod_{i=1}^\ell p_i(X)\right]\right)^{1/2}.$$
\end{defn}

Note that by Cauchy-Schwartz we have that $M_\ell(p)\leq \E[p(X)^2]^{1/4}$.  We now define a similar quantity more closely related to what is used in \cite{Moments}.

\begin{defn}
Let $q:(\R^n)^d\rightarrow \R$ by a degree-$d$ multilinear form.  Let $X^i, 1\leq i \leq d$ be independent standard $n$-dimensional Gaussians.  For integers $1\leq \ell \leq d$ define $M_\ell(q)$ in the following way.  Consider all possible choices of: a partition of $\{1,\ldots,d\}$ into non-empty subsets $S_1,\ldots,S_\ell$, with $S_i = \{c_{i,1},\ldots,c_{i,d_i}\}$; and a set of multilinear forms $q_i$ of degree-$d_i$ with $\E[q_i(X^{c_{i,1}},\ldots,X^{c_{i,d_i}})^2] \leq 1$.  We define $M_\ell(q)$ to be the supremum over all such choices of $S_i$ and $q_i$ of
$$
\left(\E\left[q(X^1,\ldots,X^d)\prod_{i=1}^\ell q_i(X^{c_{i,1}},\ldots,X^{c_{i,d_i}})\right]\right)^{1/2}.
$$
\end{defn}

We now state the moment bound whose proof will take up the rest of this Section.

\begin{prop}\label{momentProp}
Let $p$ be a homogenous degree $d$ multilinear polynomial, and $X$ a family of independent standard Gaussians, and $k\geq 2$.  Then
$$
\E[|p(X)|^k] = \Theta_d\left( \sum_{\ell=1}^d M_\ell(p) k^{\ell/2} \right)^k.
$$
\end{prop}

This is essentially a version of Theorem 1 of \cite{Moments}:
\begin{nthm}[\cite{Moments} Theorem 1]\label{LatalaThm}
For $q$ a degree-$d$ multilinear form and $X^i$ independent standard $n$-dimensional Gaussians and $k$ an integer at least 2,
$$
\E[|q(X^1,\ldots,X^d)|^k] = \Theta_d\left( \sum_{\ell=1}^d M_\ell(q) k^{\ell/2} \right)^k.
$$
\end{nthm}
\begin{proof}[Proof of Proposition \ref{momentProp}]
The basic idea of the proof is the relate $M_\ell(p)$ to $M_\ell(q)$ and $\E[|p|^k]$ to $\E[|q|^k]$ for $q$ the symmetric multilinear form associated to a multilinear polynomial $p$.

Let $q$ be the associated symmetric multilinear form associated to $p$.  We claim that for each $\ell$ that $M_\ell(p)=\Theta_d(M_\ell(q))$.  Suppose that $p_1$ and $p_2$ are degree $d$ multilinear polynomials, and $q_1$ and $q_2$ the associated symmetric multilinear forms.  It is easy to see (by using the standard basis of coefficients) that $\E[p_1(X)p_2(X)] = d!\E[q_1(X^1,\ldots,X^d)q_2(X^1,\ldots,X^d)].$  Similarly it is easy to see that if $p$ is a degree $d$ multilinear polynomial, and $p_i$ are degree $d_i$ multilinear polynomials on distinct sets of coordinates, and $q,q_i$ their associated symmetric multilinear forms we have
\begin{align*}
\E & \left[ p(X)\prod p_i(X)\right]  = \frac{1}{\prod d_i!}\E\left[q(X)\prod q_i(X^{(i)}) \right].
\end{align*}
Where $q(X) = q(X^1,\ldots,X^d)$, and $q_i(X^{(i)}) = q_i(X^{d_1+\ldots+d_{i-1}+1},\ldots,X^{d_1+\ldots+d_{i-1}+d_i}).$
This means that $M_\ell(p) = O_d(M_\ell(q))$ since given the appropriate $S_i,d_i,p_i$ we can use the symmetrizations of the $p_i$ to get as good a bound for $M_\ell(q)$ up to a constant factor.  To show the other direction we need to show that $M_\ell(q)$ is not changed by more than a constant factor if we require that the $q_i$ are supported on disjoint sets of coordinates.  But we note that if you randomly assign each coordinate to a $q_i$ and take the part that only depends on those coordinates, you loose a factor of at most $d^d$ on average.

Hence we have that
$$
\E[|q(X^1,\ldots,X^d)|^k] = \Theta_d\left( \sum_{\ell=1}^d M_\ell(p) k^{\ell/2} \right)^k.
$$
We just need to show that the moments of $p$ to the moments of $q$ are the same up to a factor of $\Theta_d(1)^k$.  This can be shown using the main Theorem of \cite{tails} which in our case states that there is some constant $C_d$ depending only on $d$ so that for any such $p,q$ and $x$,
$$
\pr(|p(X)|>x) \leq C_d\pr(|q(X^1,\ldots,X^d)|>x/C_d)
$$
and
$$
\pr(|q(X^1,\ldots,X^d)|>x) \leq C_d\pr(|p(X)|>x/C_d).
$$
Our result follows from noting that for any random variable $Y$ that
$$
\E[|Y|^k] = \int_0^\infty kx^{k-1}\pr(|Y|>x)dx.
$$
\end{proof}

\section{Structure}\label{structureSec}

In this Section, we will prove the following structure theorem for degree-$d$ multilinear polynomials.

\begin{prop}\label{structureProp}
Let $p$ be a degree-$d$ multilinear polynomial where the sum of the squares of its coefficients is at most 1.  Let $m_1\leq m_2\leq \ldots \leq m_d$ be integers.  Then there exist integers $n_1,n_2,\ldots,n_d$, $n_i = O_d(m_1m_2\cdots m_{i-1})$ and non-constant, homogeneous multilinear polynomials $h_1,\ldots,h_d$, $P_{i,j}, 1\leq i\leq d, 1\leq j \leq n_i$ so that:
\begin{enumerate}
\item $h_i$ is degree $i$
\item If $P_{i,a_1}\cdots P_{i,a_i}$ appears as a term in $h_i(P_{i,j})$, then the sum of the degrees of the $P_{i,a_i}$ is $d$
\item The sum of the squares of the coefficients of $h_i$ is $O_d(1)$
\item The sum of the squares of the coefficients of $P_{i,j}$ is 1
\item Each variable occurs in at most one monomial in $h_i$
\item If $Y$ is a standard Gaussian and $k\leq m_i$ then $\E[|P_{i,j}(Y)|^k] = O_d(\sqrt{k})^k$.
\item $p(Y) = \sum_{i=1}^d h_i(P_{i,1}(Y),P_{i,2}(Y),\ldots,P_{i,n_i}(Y)).$
\end{enumerate}
\end{prop}

This will allow us to write $p$ in terms of other polynomials each with smaller moments.  The basic idea of the proof follows from a proper interpretation of Proposition \ref{momentProp}.  Essentially Proposition \ref{momentProp} says that the higher moments of $p$ will be small unless $p$ has some significant component consisting of a product of polynomials $P_1,\ldots,P_\ell$ of lower degree.  The basic idea is that if such polynomials exist, we can split off these $P_i$ as new polynomials in our decomposition, leaving $p-P_1\cdots P_\ell$ with smaller size than $p$.  We repeatedly apply this procedure to $p$ and all of the other polynomials that show up in our decomposition.  Since each step decreases the size of the polynomial being decomposed, and produces only new polynomials of smaller degree, this process will eventually terminate.  Beyond these ideas, the proof consists largely of bookkeeping to ensure that we have the correct number of $P$'s and that they have an appropriate number of small moments.

\begin{proof}
We first prove our statement for homogeneous, multilinear polynomials $p$.  We reduce the general case to this one by writing $p$ as a sum of its homogeneous parts and decomposing each of them.

We would like to simply use the decomposition $P_{1,1}=p$ and $h_1$ is the identity, but the moments of $p$ may be too large.  On the other hand, we know by Proposition \ref{momentProp} that this can only be the case if $p$ has large correlation with some product of smaller degree polynomials $P_1\cdots P_k$.  So if $c=\E[p\cdot P_1\cdots P_k]$, we can write $p'=p-c P_1\cdots P_k$.  Now either $p'$ has small moments or we can break off another product of polynomials.  This process must eventually terminate because when we replaced $p$ by $p'$ we decreased the expectation of its square by $c^2$.  We will then apply this technique recursively to each of the $P_i$.

We define a dot product on the space of multilinear polynomials $\dotp{P}{Q} = \E[P(Y)Q(Y)]$ where $Y$ is a standard Gaussian.  Note that the square of the corresponding norm is just $|P|^2$ equals the sum of the squares of the coefficients of $P$.

We begin by letting $q=p$.  We note that by Proposition \ref{momentProp} that the $k^{th}$ moment of $q$ for $k\leq m_1$ is $O_d(\sqrt k)^k$ unless for some $2\leq \ell \leq d$ we have that $M_\ell(q) \geq m_1^{m_1/2}/m_1^{\ell/2}$, or equivalently, unless there exist polynomials $P_1,\ldots,P_\ell$ of norm 1, so that $c=\dotp{q}{P_1\cdots P_\ell} \geq m_1^{(1-\ell)/2}$.  If this is the case, we replace $q$ by $q'=q-cP_1\cdots P_\ell$.  Note that $|q'|^2 = |q|^2-c^2$.  We repeat this process with $q'$ until finally we are left with a polynomial $q$ so that for all $k\leq m_1$ the $k^{th}$ moment of $q$ is $O_d(\sqrt k)^k$ (this process must terminate since at each step we decrease $|q|^2$ by at least $m_1^{1-d}$).  We now can write $p$ as $q$ plus a sum of $c_i$ times products of lower degree polynomials.  It should be noted that the sum of the squares of the $c_i$ is at most 1.  Letting $P_{1,1}=q$ and $h_1$ be the identity, we can now write
$$
p(Y) = \sum_{i=1}^d h_i(P_{i,1}(Y),P_{i,2}(Y),\ldots,P_{i,n_i}(Y)).
$$
Where $|h_i|=O_d(1)$, $|P_{i,j}| \leq 1$, $n_i = O_d(m_1^{i-1})$, and for $k\leq m_1$, the $k^{th}$ moment of $P_{1,j}$ is $O_d(\sqrt k)^k.$  Unfortunately, the moments of the other $P$'s might be too large.  We show by induction on $s$ that we have such a decomposition where all of the appropriate moments of the $P_{i,j}$ for $i\leq s$ are bounded and so that $n_i = O_d(m_1m_2\cdots m_{i-1})$ for all $i$.

We have already proved the $s=1$ case.  To prove the general case, we first write $p$ as $\sum_{i=1}^d h_i(P_{i,1}(Y),P_{i,2}(Y),\ldots,P_{i,n_i}(Y))$ using the induction hypothesis.  This satisfies all of our criteria except that the $P_{s,j}$ might have moments which are too large.  We fix this by rewriting each of the $P_{s,j}$ using the same method we originally used to rewrite $p$, only guaranteing that the first $m_s$ moments are small.  This will make it so that our new $P_{s,j}$ have appropriately bounded moments, but may introduce new terms in the $h_t$ for $t>s$ (if some term shows up in multiple monomials, define several $P_{i,t}$ that are equal).
We need to make sure that we did not introduce too many new terms and that the sum of the squares of the coefficients is not too large.

To show the latter note that our original procedure at most doubled the sum of the squares of the coefficients.  Therefore applying this to each $P_i$ in a term $cP_1\cdots P_s$ will increase the sum of the squares of the coefficients by a factor of at most $2^s$.  Hence since the sum of the squares of the coefficients was $O_d(1)$ before, it still is afterwards.

Finally we need to show that our new decomposition did not introduce too many new terms.  It is not hard to see that for each $P_{s,i}$ we need to introduce $O_d(m_s^{t-s})$ new $P_{t,j}$ terms.  Therefore the total number of such new terms is $O(m_s^{t-s}n_s) = O_d(m_1m_2\cdots m_{t-1})$.

Finally we note that our induction terminates at $s=d$.  This is because the $P_{d,j}$ must be linear polynomials of bounded norm, and therefore automatically satisfy the necessary moment bounds.  This completes our inductive step and proves the Proposition.
\end{proof}

\section{FT-Mollification}\label{FTMSec}

We let $F$ be a degree-$d$ polynomial threshold function $F=\sgn(p)$, where $p$ is a degree $d$ multilinear polynomial in $n$ variables whose sum of squares of coefficients equals 1.  We pick $m_1,\ldots,m_d$ (their exact sizes will be determined later).  For later convenience, we assume the $m_i$ are all even.  We then have a decomposition of $F$ given by Proposition \ref{structureProp} as
\begin{align*}
F(X) & = \sgn\left(\sum_{i=1}^d h_i(P_{i,1}(X),\ldots,P_{i,n_i}(X)) \right)\\ & = f(P_1(X),P_2(X),\ldots,P_d(X)) \\ & = f(P(X))
\end{align*}
where $P_i(X)$ is the vector-valued polynomial $(P_{i,1}(X),\ldots,P_{i,n_i}(X))$, $P$ is the vector of all of them, and $f$ is the function $f(P_1,\ldots,P_d) = \sgn(\sum h_i(P_i))$.  Furthermore, we have that for $k\leq d m_i$ the $k^{th}$ moment of any coordinate of any coordinate of $P_i$ is $O_d(\sqrt{k})^k$.  We also have that $h_i$ is a degree $i$ multilinear polynomial the sum of the squares of whose coefficients is at most 1.

Our basic strategy now will involve approximating $f$ by a smooth function $\tilde f$, and letting $\tilde F(X) = \tilde f(P(X))$.  We will then proceed to prove
\begin{equation}\label{ApproxEqn}
\E[F(Y)] \approx_\epsilon \E[\tilde F(Y)] \approx_\epsilon \E[\tilde F(X)] \approx_\epsilon \E[F(X)].
\end{equation}
We will produce $\tilde f$ from $f$ using the technique of mollification.  Namely we will have $\tilde f = f * \rho$ for an appropriately chosen smooth function $\rho$.  However, we will need this $\rho$ to have several other properties so we will go into some depth here to construct it.

\begin{lem}\label{rhoLem}
Given an integer $n\geq 0$ and a constant $C$, there is a function $\rho_C:\R^n\rightarrow \R$ so that
\begin{enumerate}
\item $\rho_C \geq 0$.
\item $\int_{\R^n} \rho_C(x)dx = 1$.
\item For any unit vector $v\in \R^n$, and any non-negative integer $k$, $\int_{\R^n} |D_v^k \rho_C(x)|dx \leq C^k$, where $D_v^k$ is the $k^{th}$ directional derivative in the direction $v$.
\item For $D>0$, $\int_{|x|>D} |\rho(x)|dx = O\left(\left(\frac{n}{CD}\right)^2\right)$.
\end{enumerate}
\end{lem}
\begin{proof}
We prove this for $C=2$ and we note that we can obtain other values of $C$ by setting $\rho_C(x) = (C/2)^n \rho_2(Cx/2).$  We begin by defining
$$
B(\xi) = \begin{cases} 1 - |\xi|^2 \ \ & \textrm{if} \ |\xi|\leq 1\\ 0 \ \ & \textrm{else} \end{cases}%
$$
We then define
$$
\rho_2(x) = \rho(x) = \frac{|\hat{B}(x)|^2}{|B|_2^2}.
$$
Where $\hat B$ denotes the Fourier transform of $B$.  Clearly $\rho$ is non-negative.  Also clearly
$$
\int_{\R^n} \rho(x)dx = \frac{|\hat B|_2^2}{|B|_2^2} = 1
$$
by the Plancherel Theorem.

For the third property we note that
$$
D_v^k \rho = \frac{1}{|B|_2^2} \sum_{i=0}^k \binom{k}{i} D_v^i(\hat B) D_v^{k-i} \overline{(\hat B)}.
$$

Letting $\xi$ be the dual vector corresponding to $v$ we have that
\begin{align*}
\left| D_v^k \rho\right|_1 & \leq \frac{1}{|B|_2^2}\sum_{i=0}^k \binom ki \left|D_v^i(\hat B) D_v^{k-i} \overline{(\hat B)}\right|_1\\
& \leq \frac{1}{|B|_2^2}\sum_{i=0}^k \binom ki \left|D_v^i(\hat B)\right|_2\left|D_v^{k-i}(\hat B)\right|_2\\
& \leq \frac{1}{|B|_2^2}\sum_{i=0}^k \binom ki \left|\xi^i B\right|_2\left|\xi^{k-i} B \right|_2\\
& \leq \frac{1}{|B|_2^2}\sum_{i=0}^k \binom ki |B|_2^2\\
& = \sum_{i=0}^k \binom ki \\
&= 2^k.
\end{align*}
For the last property we note that it is enough to prove that
$$
\int_{\R^n} |x|^2\rho(x) dx = O(n^2).
$$
We have that
\begin{align*}
\int_{\R^n} |x|^2\rho(x) dx & = \frac{1}{|B|_2^2} \sum_{i=1}^n |x_i\hat B|_2^2\\
& = \frac{\sum_{i=1}^n\left| \frac{\partial B}{\partial \xi_i}\right|_2^2}{|B|_2^2}.
\end{align*}
Now $ \frac{\partial B}{\partial \xi_i} $ is $2\xi_i$ on the unit ball and 0 outside.  Hence the sum of the squares of these is $2|\xi|^2$ on $|\xi|<1$ and 0 outside.  Hence since both numerator and denominator above are integrals of spherically symmetric functions, their ratio is equal to
\begin{align*}
\int_{\R^n} |x|^2\rho(x) dx & = \frac{2\int_0^1 r^{n+1} dr}{\int_0^1 r^{n-1}(1-r^2)^2 dr}.
\end{align*}
Using integration by parts, the denominator is
\begin{align*}
\int_0^1 r^{n-1}(1-r^2)^2 dr & = \frac{4}{n}\int_0^1 r^{n+1}(1-r^2)dr \\ & = \frac{16}{n(n+2)}\int_0^1 r^{n+3}dr \\ & = \frac{16}{n(n+2)(n+4)}.
\end{align*}
Hence
$$
\int_{\R^n} |x|^2\rho(x) dx = \frac{n(n+4)}{8} = O(n^2).
$$
\end{proof}

We are now prepared to define $\tilde f$.  We pick constants $C_1,\ldots,C_d$ (to be determined later).  We let \begin{equation}\label{rhoDefEqn}\rho(P_1,\ldots,P_d) = \rho_{C_1}(P_1)\cdot \rho_{C_2}(P_2)\cdots \rho_{C_d}(P_d).\end{equation}  Above the $\rho_{C_i}$ is defined on $\R^{n_i}$. We let $\tilde f$ be the convolution $\tilde f = f * \rho$.

\section{Taylor Error}\label{TaylorSec}

In this Section, we prove the middle approximation of Equation \ref{ApproxEqn} for appropriately large $k$.  The basic idea will be to approximate $\tilde f$ by its Taylor series, $T$.  $T(P(X))$ will be a polynomial of degree at most $k$ and hence $\E[T(P(Y))]=\E[T(P(X))].$  Furthermore, we will bound the Taylor error by some polynomial $R$ and show that $\E[R(P(Y))]=\E[R(P(X))]$ is $O(\epsilon)$ for appropriate choices of $m_i,C_i$.  In particular, we let $T$ be the polynomial consisting of all of the terms of the Taylor expansion of $\tilde f$ whose total degree in the $P_i$ coordinates is less than $m_i$ for all $i$.  Note that a polynomial of this form is about the best we can do since we only have control over the size of moments up to the $m_i^{th}$ moment on the $i^{th}$ block of coordinates.  Our error bound will be the following

\begin{prop}\label{taylorErrorProp}
$$
|T(P) - \tilde f(P)| \leq \prod_{i=1}^d \left(1+\frac{C_i^{m_i}|P_i|^{m_i}}{m_i!} \right)-1.
$$
\end{prop}

First we prove a Lemma dealing with Taylor error for a single batch of coordinates,
\begin{lem}\label{basicTaylorLem}
If $g$ is a multivariate function, $\tilde g= g * \rho_C$ and $T$ is the polynomial consisting of all terms in the Taylor expansion of $\tilde g$ is degree less than $m$, then
$$
|g(x)-T(x)| \leq \frac{|g|_\infty C^m |x|^m}{m!}.
$$
\end{lem}
\begin{proof}
Let $v$ be the unit vector in the direction of $x$.  Let $L$ be the line through $0$ and $x$.  We note that the restriction of $T$ to $L$ is the same as the first $m-1$ terms of the Taylor series for $\tilde g|_L$.  Using standard error bounds for Taylor polynomials we find that
$$
|g(x)-T(x)| \leq \frac{|D_v^m \tilde g|_\infty |x|^m}{m!}.
$$
But
\begin{align*}
|D_v^m \tilde g|_\infty & = |g * D_v^m \rho_C |_\infty\\
& \leq |g|_\infty |D_v^m \rho_C|_1\\
& \leq |g|_\infty C^m.
\end{align*}
Plugging this in yields our result.
\end{proof}

\begin{proof}[Proof of Proposition \ref{taylorErrorProp}]
The basic idea of the proof will be to repeatedly apply Lemma \ref{basicTaylorLem} to one batch of coordinates at a time.  We begin by defining some operators on the space of bounded functions on $\R^{n_1}\times\R^{n_2}\times\cdots\times\R^{n_d}$.  For such $g$, define $g^{\tilde i}$ to be the convolution of $g$ with $\rho_{C_i}$ along the $i^{th}$ set of coordinates.  Define $g^{T_i}$ to be the Taylor polynomial in the $i^{th}$ set of variables of $g^{\tilde i}$ obtained by taking all terms of total degree less than $m_i$.  Note that for $i\neq j$ the operations $\tilde i$ and $T_i$ commute with the operations $\tilde j$ and $T_j$ since they operate on disjoint sets of coordinates.  Note that $\tilde f = f^{\tilde 1 \tilde 2\cdots \tilde d}$ and $T=f^{T_1 T_2 \cdots T_d}$.  For $1\leq i \leq d$ let $f_i = f^{\tilde 1 \tilde 2\cdots \tilde i}$ and $T_i = f^{T_1 T_2 \cdots T_i}$.

We prove by induction on $s$ that
$$
|T_s(P)-f_s(P)| \leq  \prod_{i=1}^s \left(1+\frac{C_i^{m_i}|P_i|^{m_i}}{m_i!} \right)-1.
$$
As a base case, we note that the $s=0$ case of this is trivial.

Assume that
$$
|T_s(P)-f_s(P)| \leq  \prod_{i=1}^s \left(1+\frac{C_i^{m_i}|P_i|^{m_i}}{m_i!} \right)-1.
$$
We have that
\begin{align*}
|T_{s+1}&(P)- f_{s+1}(P)|\\ & \leq |T_s^{T_{s+1}}(P)-T_s^{\tilde{s+1}}(P)| + |T_s^{\tilde{s+1}}(P)-f_{s}^{\tilde{s+1}}(P)|.
\end{align*}
Note that
$$
T_s^{\tilde{s+1}}(P)-f_{s}^{\tilde{s+1}}(P) = \left(T_s - f_s\right)^{\tilde{s+1}}(P).
$$
Therefore since $\tilde{s+1}$ involves only convolution with a function of $L^1$ norm 1 we have that
$$
|T_s^{\tilde{s+1}}(P)-f_{s}^{\tilde{s+1}}(P)| \leq |T_s(P)-f_s(P)|_{\infty,s+1}
$$
where the subscript denotes the $L^\infty$ norm over just the $s+1^{st}$ set of coordinates.  By the inductive hypothesis, this is at most
$$
\prod_{i=1}^s \left(1+\frac{C_i^{m_i}|P_i|^{m_i}}{m_i!} \right)-1.
$$

On the other hand, applying Lemma \ref{basicTaylorLem} we have that
$$
|T_s^{T_{s+1}}(P)-T_s^{\tilde{s+1}}(P)| \leq \frac{C_{s+1}^{m_{s+1}}|P_{s+1}|^{m_{s+1}}}{m_{s+1}!}\left|T_s \right|_{\infty,s+1}.
$$
By the inductive hypothesis,
\begin{align*}
\left|T_s \right|_{\infty,s+1} & \leq |f_s|_{\infty,s+1} + |T_s-f_s|_{\infty,s+1} \\ & \leq \prod_{i=1}^s \left(1+\frac{C_i^{m_i}|P_i|^{m_i}}{m_i!} \right).
\end{align*}

Combining the above bounds, we find that
\begin{align*}
|T_{s+1}-f_{s+1}| \leq &  \prod_{i=1}^s \left(1+\frac{C_i^{m_i}|P_i|^{m_i}}{m_i!} \right) - 1 \\ & + \frac{C_{s+1}^{m_{s+1}}|P_{s+1}|^{m_{s+1}}}{m_{s+1}!}\prod_{i=1}^s \left(1+\frac{C_i^{m_i}|P_i|^{m_i}}{m_i!} \right)\\
= & \prod_{i=1}^{s+1} \left(1+\frac{C_i^{m_i}|P_i|^{m_i}}{m_i!} \right)-1.
\end{align*}
\end{proof}

We can now prove the desired approximation result
\begin{prop}\label{polyapproxErrorProp}
If $\tilde F,P,T$ as above with $m_i= \Omega_d(n_i C_i^2)$, $m_i \geq \log(2^d/\epsilon)$ for all $i$, and if $k\geq dm_i$ for all $i$, then for $X$ and $Y$ are $k$-independent families of standard Gaussians,
$$
\E[\tilde F(Y)] \approx_\epsilon \E[\tilde F(X)].
$$
\end{prop}
\begin{proof}
We note that since $T\circ P$ is a polynomial of degree at most $k$ we have that $\E[T(P(X))] = \E[T(P(Y))].$  Hence, it suffices to show that
$$
\E[|F-T|(P(X))],\E[|F-T|(P(Y))] = O(\epsilon).
$$
We will show this only for $X$ as $Y$ is analogous.  By Proposition \ref{taylorErrorProp} we have that $|F-T|$ is bounded by
$$
\prod_{i=1}^d \left(1+\frac{C_i^{m_i}|P_i|^{m_i}}{m_i!} \right)-1.
$$
This is a sum over non-empty subsets $S\subseteq\{1,2,\ldots,d\}$ of
$$
\prod_{i\in S} \left(\frac{C_i^{m_i}|P_i|^{m_i}}{m_i!} \right).
$$
Since there are only $2^d-1$ such $S$, it is enough to show that each term individually has expectation $O(\epsilon/2^d)$.  On the other hand, we have by AM-GM that each term is at most
$$
\frac{1}{|S|}\sum_{i\in S} \left(\frac{C_i^{m_i}|P_i|^{m_i}}{m_i!} \right)^{|S|}.
$$
Now the expectation of $|P_i|^{m_i|S|}$ is at most $n_i^{m_i|S|}$ times the average of the $m_i|S|^{th}$ moments of the coordinates of $P_i$.  These by assumption are $O_d(\sqrt{m_i|S|})^{m_i|S|}$.  There are $n_i$ coordinates so the moment of $|P_i|$ is at most $O_d(\sqrt{n_im_i|S|})^{m_i|S|}$. Hence the error is at most
\begin{align*}
O(2^d)\max_{i,s}&\left\{ O_d\left( \frac{C_i\sqrt{n_i m_i s}}{m_i}\right)^{m_i s} \right\}\\
& = O(2^d)\max_{i,s}\left\{ O_d\left( \frac{C_i\sqrt{n_i}}{\sqrt{m_i}}\right)^{m_i s} \right\}\\
& \leq O(2^d) e^{-\min_i m_i} = O(\epsilon).
\end{align*}

\end{proof}

\section{Approximation Error}\label{AntiConSeq}

In this Section, we will prove the first and third approximations in Equation \ref{ApproxEqn}.  We begin with the first, namely
$$
\E[F(Y)] \approx_\epsilon \E[\tilde F(Y)].
$$
Our basic strategy will be to bound
$$
|\E[F(Y)] - \E[\tilde F(Y)] \leq \E[|F(Y)-\tilde F(Y)|].
$$
In order to get a bound on this we will first show that $F-\tilde F$ is small except where $p(Y)$ is small, and then use anti-concentration results to show that this happens with small probability.  This will be true because $\rho$ is small away from $0$.  We begin by proving a Lemma to this effect.

\begin{lem}\label{rhoConcentrationLem}
Let $\rho$ be the function defined in Equation \ref{rhoDefEqn}.  Then for any $D>0$ we have that
$$
\int_{\substack{(x_1,\ldots,x_d)\in \R^{n_1}\times\cdots\times\R^{n_d} \\ \exists i: |x_i| > Dn_i\sqrt{d}/C_i}} |\rho(x)|dx = O(D^{-2})
$$
\end{lem}
This will hold essentially because of the concentration property held by each $\rho_{C_i}$.
\begin{proof}
We integrate over the region where $|x_i| > Dn_i\sqrt{d}/C_i$ for each $i$.  This is a product over $j\neq i$ of $\int_{\R^{n_j}} \rho_{C_j}(x_j)$ times $\int_{|x|>Dn_i\sqrt{d}/C_i} |\rho(x)|dx$.  By Lemma \ref{rhoConcentrationLem} the former integrals are all 1, and the latter is $O(D^{-2}/d)$.  Summing over all possible $i$ yields $O(D^{-2})$.
\end{proof}

Recall that $f$ was $\sgn\circ h$, where $h=\sum h_i$ given in the decomposition of $p$ from Proposition \ref{structureProp}.  Recall that $\tilde f = f * \rho$.  We want to bound the error in approximating $f$ by $\tilde f$.  The following, is a direct consequence of Lemma \ref{rhoConcentrationLem}.

\begin{lem}\label{smoothErrorBoundLem}
Suppose $x=(x_1,\ldots,x_d)\in \R^{n_1}\times\cdots\times\R^{n_d}$.  Suppose also that for some $D>0$ and for all $y=(y_1,\ldots,y_d)\in\R^{n_1}\times\cdots\times\R^{n_d}$ so that $|x_i-y_i|\leq Dn_i\sqrt{d}/C_i$ that $h(x)$ and $h(y)$ have the same sign, then
$$
|f(x)-\tilde f(x)| = O\left(\min\{ 1 , D^{-2} \}\right).
$$
\end{lem}
\begin{proof}
To show that the error is $O(1)$, we note that since $\rho\geq 0$ and $\int \rho(x)dx=1$ that $\tilde f(x) = (f * \rho)(x) \in [\inf(f),\sup(f)] \subseteq [-1,1]$.  Therefore $|f-\tilde f| \leq |f| + |\tilde f| \leq 2.$

For the latter, we note that $\tilde f(x) = \int_y f(y)\rho(x-y) dy$.  We note that since the total integral of $\rho$ is 1 that
$$
f(x) - \tilde f (x) = \int_y (f(x)-f(y))\rho(x-y) dy.
$$
We note that by assumption unless $|x_i-y_i| > Dn_i\sqrt{d}/C_i$ for some $i$ that the integrand is 0.  But outside of this, the integrand is at most $2\rho(x-y)$.  By Lemma \ref{rhoConcentrationLem} the total integral of this is $O(D^{-2})$.
\end{proof}

We now know that $f$ is near $\tilde f$ at points $x$ not near the boundary between the $+1$ and $-1$ regions.  Since we cannot directly control the size of these regions, we want to relate this to the region where $|h(x)|$ is small.  This should work since unless $x$ is very large, $h$ will have derivatives that aren't too big.  In particular, we prove the following.

\begin{lem}\label{near0boundLem}
Let $x\in \R^n$.  Suppose that we have $B_i\geq 0$ so that $|P_{i,j}(x)| \leq B_i$ for all $i,j$.  We have that $|F(x) - \tilde F(x)|$ is at most the minimum of $O(1)$ and
\begin{align*}
&O_d\left(\max\left\{\left(\frac{|p(x)|}{\sum_{i=1}^d n_i^2 B_i^{i-1}/C_i} \right)^{-2} ,\left(\frac{B_iC_i}{n_i}\right)^{-2}  \right\} \right).
\end{align*}
\end{lem}
\begin{proof}
The bound of $O(1)$ follows immediately from Lemma \ref{smoothErrorBoundLem}.  For the other bound, let
$$
D = \min\left\{\frac{|p(x)|}{d2^d\sum_{i=1}^d n_i^2 B_i^{i-1}/C_i},\min_i\left\{\frac{B_iC_i}{n_i\sqrt{d}}\right\} \right\}.
$$
By Lemma \ref{smoothErrorBoundLem}, it suffices to show that for any $Q=(Q_1,\ldots,Q_n)\in\R^{n_1}\times\cdots\times\R^{n_d}$ so that $|Q_i-P_i(x)|\leq Dn_i\sqrt{d}/C_i$ that $h(P(x))=p(x)$ and $h(Q)$ have the same sign.  To do this, we write $h=h_1+\cdots+h_d$ and we note that
$$
|h(P(x))-h(Q)| \leq \sum_{i=1}^d |P_i(x)-Q_i||h_i'(z)|.
$$
Where $h_i'(z)$ is the directional derivative of $h_i$ in the direction from $P_i(x)$ to $Q_i$, and $z$ is some point along this line.  First, note that $|Q_i-P_i(x)|\leq B_i$.  Therefore, each coordinate of $z$ is at most $2B_i$.  Note that $h_i$ is a sum of at most $n_i$ monomials of degree $i$ with coefficients at most 1.  The derivative of each monomial at $z$ is at most $\sqrt{d}2^d B_i^{i-1}$.  Therefore, $|h_i'(z)|\leq \sqrt{d}2^d n_i B_i^{i-1}$.  Therefore,
\begin{align*}
|h(P(x))-h(Q)| & \leq \sum_{i=1}^d |P_i(x)-Q_i||h_i'(z)|\\
& \leq \sum_{i=1}^d (D n_i \sqrt{d} / C_i )(\sqrt{d}2^d n_i B_i^{i-1})\\
& \leq D \sum_{i=1}^d d2^d n_i^2 B_i^{i-1}/C_i\\
& \leq |h(P(x))|.
\end{align*}
Therefore $h(P(x))$ and $h(Q)$ have the same sign, so our bound follows by Lemma \ref{smoothErrorBoundLem}.
\end{proof}

We take this bound on the approximation error and prove the following Lemma on the error of expectations.

\begin{lem}\label{expectationErrorLem}
Let $Z$ be a random variable valued in $\R^n$.  Let $B_i>1$ be real numbers.  Let $M=\sum_{i=1}^d n_i^2 B_i^{i-1}/C_i$.  Then
\begin{align*}
& |\E[F(Z)] - \E[\tilde F(Z)]|  =\\
& O_d( \pr(\exists i,j:|P_{i,j}(Z)| > B_i)  +  M + \pr(|p(Z)|\leq \sqrt{M})).
\end{align*}
Furthermore,
\begin{align*}
\E[F(Z)] \leq & \E[\tilde F(Z)] \\ & + O_d( \pr(\exists i,j :|P_{i,j}(Z)| > B_i)  +  M) \\ & + 2\pr(-\sqrt{M}<p(Z)<0),
\end{align*}
and
\begin{align*}
\E[F(Z)] \geq & \E[\tilde F(Z)] \\ & + O_d( \pr(\exists i,j:|P_{i,j}(Z)| > B_i)  +  M)\\ & - 2\pr(0<p(Z)<\sqrt{M}).
\end{align*}
\end{lem}
\begin{proof}
We note that $|F(Z) - \tilde F(Z)| = O(1)$.  Also note that $\frac{1}{M} \leq \frac{B_iC_i}{n_i}$ for all $i$.  The first inequality follows by noting that Lemma \ref{near0boundLem} implies that unless $|P_{i,j}(Z)|>B_i$ for some $i,j$ that the following hold:
\begin{enumerate}
\item If $|p(z)|< \sqrt{M}$, $|F(Z) - \tilde F(Z)| \leq 2$.
\item If $|p(z)| \geq \sqrt{M}$, $|F(Z) - \tilde F(Z)| = O_d(M)$.
\end{enumerate}
The other two inequalities follow from noting that if $p(z)<0$, then $F(Z)\leq\tilde F(Z)$ and if $p(Z)>0$ then $F(Z)\geq \tilde F(Z)$.
\end{proof}

We are almost ready to prove the first of our approximation results, but we first need a theorem on the anticoncentration of Gaussian polynomials.  In particular a consequence of \cite{anticoncentration} Theorem 8 is:

\begin{thm}[Carbery and Wright]\label{anticoncentrationTheorem}
Let $p$ be a degree $d$ polynomial, and $Y$ a standard Gaussian.  Suppose that $\E[p(Y)^2]=1$.  Then, for $\epsilon>0$,
$$
\pr(|p(Y)|<\epsilon) = O(d \epsilon^{1/d}).
$$
\end{thm}

We are now prepared to prove our approximation result.

\begin{prop}\label{yanitconcentrationErrorProp}
Let $p,F,\tilde F,h,m_i,n_i,C_i$ be as above and let $\epsilon>0$.  Let $B_i = \Omega_d (\sqrt{\log(n_i/\epsilon)})$ be some real numbers.  Suppose that $m_i > B_i^2$ and that $C_i = \Omega_d (n_i^2 B_i^{i-1} \epsilon^{-2d})$ for all $i$.  Then, if the implied constants for the bounds on $B_i$ and $C_i$ are large enough,
$$
|\E[F(Y)] - \E[\tilde F(Y)]| = O(\epsilon).
$$
\end{prop}
\begin{proof}
We bound the error using Lemma \ref{expectationErrorLem}.  We note that the probability that $|P_{i,j}(Y)|\geq B_i$ can be bounded by looking at the $\log(d n_i/\epsilon)=k^{th}$ moment, yielding a probability of $\frac{O_d(\sqrt{k})^k}{B_i^k} \leq e^{-k} = \frac{\epsilon}{dn_i}$.  Taking a union bound over all $j$ gives a probability of $\frac{\epsilon}{d}$.  Taking a union bound over $i$ yields a probability of at most $\epsilon$.

Next we note that
$$
M = \sum_{i=1}^d n_i^2 B_i^{i-1}/ C_i = O_d(\epsilon^{2d}).
$$
Hence if our constants were chosen to be large enough, by Theorem \ref{anticoncentrationTheorem}
$$
\pr(|p(Y)|<\sqrt{M}) = O(\epsilon).
$$
This proves our result.
\end{proof}

If we could prove Proposition \ref{yanitconcentrationErrorProp} for $X$ instead of $Y$, we would be done.  Unfortunately, Theorem \ref{anticoncentrationTheorem} does not immediately apply for families that are merely $k$-independent.  Fortunately, we can work around this to prove Proposition \ref{mainProp}.  In particular, we will use the inequality versions of Lemma \ref{expectationErrorLem} to obtain upper and lower bounds on $\E[F(X)]$ in terms of $\E[\sgn(p(Y)+c)]$, and make use of anticoncentration for $p(Y)$.

\begin{proof}[Proof of Proposition \ref{mainProp}]
Let $B_i = \Omega_d(\sqrt{\log(1/\epsilon)})$ with sufficiently large constants.  Define $m_i$ and $C_i$ so that $C_i = \Omega_d\left(\left(\prod_{j=1}^{i-1}m_j\right)^2 B_i^{i-1} \epsilon^{-2d}\right)$ and $m_i \geq \Omega_d\left(\left(\prod_{j=1}^{i-1}m_j\right)C_i^2\right),\log(2^d/\epsilon),B_i^2$, all with sufficiently large constants.  Note that this is achievable by setting $C_i = \Omega_d\left( \epsilon^{-7^i d}\right)$, $m_i = \Omega_d\left(\epsilon^{-3\cdot7^i d} \right)$.  Let $k=d\max_i m_i$.  Note $k$ can be as small as $O_d(\epsilon^{-4d\cdot 7^d})$.  Using these parameters, define $n_i,h_i,P_{i,j},f,\tilde f,\tilde F$ as described above.  Note that since $n_i = O_d\left(\prod_{j=1}^{i-1} m_i\right)$ that $C_i = \Omega_d(n_i^2 B_i^{i-1} \epsilon^{-2d})$ and $m_i = \Omega_d (n_i C_i^2)$.  Therefore, for $Y$ a family of independent standard Gaussians and $X$ a family of $k$-independent standard Gaussians, Propositions \ref{taylorErrorProp} and \ref{yanitconcentrationErrorProp} imply that
$$
\E[F(Y)] \approx_\epsilon \E[\tilde F(Y)] \approx_\epsilon \E[\tilde F(X)].
$$
We note that the $M$ in Lemma \ref{expectationErrorLem} is $O_d(\epsilon^{2d})$ with sufficiently small constant.  Therefore, by Lemma \ref{expectationErrorLem} $|\E[F(X)] -  \E[\tilde F(X)]|$ is at most
\begin{align*}
O(\epsilon) + 2\pr(|p(X)|<O_d(\epsilon^{d})) + \pr(\exists i,j:|P_{i,j}(X)| > B_i ).
\end{align*}
We note that by looking at the $\log(d n_i/\epsilon)$ moments of the $P_{i,j}$ that the last probability is $O(\epsilon)$.  Therefore, combining this with the above we get that
$$
\E[F(X)] \geq \E[F(Y)]+O(\epsilon) - 2\pr(0<p(X)<O_d(\epsilon^d)),
$$
and
$$
\E[F(X)] \leq \E[F(Y)]+O(\epsilon) + 2\pr(-O_d(\epsilon^d)<p(X)<0).
$$
But this implies that
$$
\E[\sgn(p(X)-O_d(\epsilon^d))] \leq \E[F(Y)] +O(\epsilon),
$$
and
$$
\E[\sgn(p(X)+O_d(\epsilon^d))] \geq \E[F(Y)] +O(\epsilon).
$$
On the other hand, applying to above to the polynomials $p\pm O_d(\epsilon^d)$,
\begin{align*}
\E[\sgn(p(Y)-O_d(\epsilon^d))] & +O(\epsilon) \leq \E[F(X)] \\ & \leq \E[\sgn(p(Y)+O_d(\epsilon^d))]+O(\epsilon).
\end{align*}
But we have that
$$
\E[\sgn(p(Y)-O_d(\epsilon^d))]\leq \E[F(Y)] \leq \E[\sgn(p(Y)+O_d(\epsilon^d))].
$$
Furthermore, $\sgn(p(Y)-O_d(\epsilon^d))$ and $\sgn(p(Y)+O_d(\epsilon^d))$ differ by at most 2, and only when $|p(Y)| = O_d(\epsilon^d)$.  By Theorem \ref{anticoncentrationTheorem}, this happens with probability $O_d(\epsilon)$.  Therefore, we have that all of the expectations above are within $O_d(\epsilon)$ of $\E[F(Y)]$, and hence $\E[F(X)] = \E[F(Y)] + O_d(\epsilon)$.  Decreasing the value of $\epsilon$ by a factor depending only on $d$ (and increasing $k$ by a corresponding factor) yields our result.
\end{proof}

\section{General Polynomials}\label{ReductionSec}

We have proved our Theorem for multilinear polynomials, but would like to extend it to general polynomials.  Our basic idea will be to show that a general polynomial is approximated by a multilinear polynomial in perhaps more variables.

\begin{lem}\label{multilinearLem}
Let $p$ be a degree $d$ polynomial and $\delta>0$.  Then there exists a multilinear degree $d$ polynomial $p_\delta$ (in perhaps a greater number of variables) so that for every $k$-independent family of random Gaussians $X$, there is a (correlated) $k$-independent family of random Gaussians $\tilde X$ so that
$$
\pr(|p(X)-p_\delta(\tilde X)|>\delta) < \delta.
$$
\end{lem}
\begin{proof}
We will pick some large integer $N$ (how large we will say later).  If $X=(X_1,\ldots,X_n)$, we let $\tilde X = (X_{i,j}), 1\leq i\leq n, 1\leq j\leq N$.  For fixed $i$ we let the collection of $X_{i,j}$ be the standard collection of $N$ standard Gaussians subject to the condition that $X_i = \frac{1}{\sqrt{N}}\sum_{j=1}^N X_{i,j}$.  Equivalently, $X_{i,j} = \frac{1}{\sqrt{N}}X_i + Y_{i,j}$ where the $Y_{i,j}$ are Gaussians with variance $1-1/N$ and covariance $-1/N$ with each other.

$\tilde X$ is $k$-independent because given any $i_1,\ldots,i_k$, $j_1,\ldots,j_k$ we can obtain the $X_{i_\ell,j_\ell}$ by first picking the $X_{i_\ell}$ randomly and independently, and picking the $Y_{i_\ell,j_\ell}$ independently of those.  But we note that this yields the same distribution we would get by setting all of the $X_{i_\ell,k}$ to be random independent Gaussians, and letting $X_i=\frac{1}{\sqrt{N}}\sum_{j=1}^N X_{i,j}$.

We now need to construct $p_\delta$ with the appropriate property.  The idea will be to replace each term $X_i^k$ in each monomial in $p$ with some degree $k$ polynomial in the $X_{i,j}$.  This will yield a multilinear degree $d$ polynomial in $\tilde X$.  We will want this new polynomial to be within $\delta'$ of $X_i^k$ with probability $1-\delta'$ for $\delta'$ some small positive number depending on $p$ and $\delta$.  This will be enough since if $\delta'<\delta/(2dn)$ the approximation will hold for all $i,k$ with probability at least $1-\delta/2$.  Furthermore with probability $1-\delta/2$, each of the $|X_i|$ will be at most $O(\log(n/\delta))$.  Therefore if this holds and each of the replacement polynomials is off by at most $\delta'$, then the value of the full polynomial will be off by at most $O(\log^d(n/\delta)\delta')$ times the sum of the coefficients of $p$.  Hence if we can achieve this for $\delta'$ small enough we are done.

Hence, we have reduced our problem to the case of $p(X) = p(X_1) = X_1^d$.  For simplicity of notation, we use $X$ instead of $X_1$ and $X_j$ instead of $X_{1,j}$.  We note that
$$
X^d = N^{-d/2}\left(\sum_{i=1}^N X_i\right)^d.
$$
Unfortunately, this is not a multilinear polynomial in the $X_i$.  Fortunately, it almost is.  Expanding it out and grouping terms based on the multiset of exponents occurring in them we find that
$$
X^d = N^{-d/2}\sum_{\substack{a_1\leq \ldots \leq a_k \\ \sum a_i=d}} \binom{d}{a_1,a_2,\ldots,a_k} \sum_{\mathcal{S}} \prod_{j=1}^k X_{i_j}^{a_j}.
$$
Where $S$ is the set of $i_1,\ldots,i_k\in\{1,\ldots,N\}$ distinct so that $i_j < i_{j+1}$ if $a_j = a_{j+1}$.
Letting $b_\ell$ be the number of $a_i$ that are equal to $\ell$ we find that this is
$$
N^{-d/2}\sum_{\substack{a_1\leq \ldots \leq a_k \\ \sum a_i=d}} \binom{d}{a_1,a_2,\ldots,a_k} \prod_\ell \frac{1}{b_\ell!}\sum_{\substack{i_1,\ldots,i_k\in [N] \\ i_j \ \textrm{ distinct}}} \prod_{j=1}^k X_{i_j}^{a_j}.
$$
Or rewriting slightly, this is
$$
\sum_{\substack{a_1\leq \ldots \leq a_k \\ \sum a_i=d}} \binom{d}{a_1,a_2,\ldots,a_k} \prod_\ell \frac{1}{b_\ell!}\sum_{\substack{i_1,\ldots,i_k\in [N] \\ i_j \ \textrm{ distinct}}} \prod_{j=1}^k \left(\frac{X_{i_j}}{\sqrt{N}}\right)^{a_j}.
$$
Now, with probability $1-\delta$, $\left|\sum_i \frac{X_i}{\sqrt{N}} \right| = O(\log(1/\delta))$.  Furthermore with probability tending to 1 as $N$ goes to infinity, $\left(\sum_i \left(\frac{X_i}{\sqrt{N}}\right)^2 \right) = 1 + O(\delta / \log^d(1/\delta)),$ and $\left(\sum_i \left(\frac{X_i}{\sqrt{N}}\right)^a \right) = O(\delta / \log^d(1/\delta))$ for each $3\leq a \leq d$.  If all of these events hold, then each term in the above with some $a_j>2$ will be $O(\delta)$, and any terms with some $a_j=2$ will be within $O(\delta)$ of
$$
\sum_{\substack{i_1,\ldots,i_k'\in \{1,\ldots,N\} \\ i_j \ \textrm{ distinct}}} \prod_{j=1}^{k'} \frac{X_{i_j}}{\sqrt{N}}
$$
where $k'$ is the largest $j$ so that $a_j=1$.  This gives a multilinear polynomial, that with probability $1-\delta$ is within $O_d(\delta)$ of $p(X)$.  Perhaps decreasing $\delta$ to deal with the constant in the $O_d$ yields our result.
\end{proof}

We can now prove Theorem \ref{mainTheorem}.

\begin{proof}[Proof of Theorem \ref{mainTheorem}]
Let $p$ be a normalized degree $d$ polynomial. Let $k$ be as required by Proposition \ref{mainProp}.  Let $Y$ be a family of independent standard Gaussians and $X$ a $k$-independent family of standard Gaussians.  Fix $\delta=(\epsilon/d)^d$.  Let $p_\delta,\tilde X, \tilde Y$ be as given by Lemma \ref{multilinearLem}.  We need to show that $\pr(p(X)>0) = \pr(p(Y)>0)+O(\epsilon).$  By construction of $p_\delta$,
$$
\pr(p(X)>0) \geq \pr(p_\delta(\tilde X)>\delta) - \delta.
$$
Applying Proposition \ref{mainProp} to the multilinear polynomial $p_\delta-\delta$, this is at least
$$
\pr(p_\delta(\tilde Y)>\delta) + O(\epsilon).
$$
Since $\tilde Y$ is $\ell$-independent for all $\ell$ (since $Y$ is), it is actually an independent family of Gaussians.  Therefore by Theorem \ref{anticoncentrationTheorem}, $\pr(|p(Y)|<\delta) = O(d\delta^{1/d}) = O(\epsilon)$.  Hence
$$
\pr(p(X)>0) \geq \pr(p_\delta(\tilde Y) > -\delta) + O(\epsilon).
$$
Noting that with probability $1-\delta$ that $p_\delta(\tilde Y)$ is at most $\delta$ less than $p(Y)$, this is at least
$$
\pr(p(Y)>0) + O(\epsilon).
$$
So
$$
\pr(p(X)>0) \geq \pr(p(Y)>0) + O(\epsilon).
$$
Similarly,
$$
\pr(p(X)<0) \geq \pr(p(Y)<0) + O(\epsilon).
$$
Combining these we clearly have
$$
\pr(p(X)>0) = \pr(p(Y)>0) + O(\epsilon)
$$
as desired.
\end{proof}

\section{Fooling PTFs of Bernoulli Random Variables}

Theorem \ref{mainTheorem} should also hold when $X$ is a $k$-independent family of Bernoulli random variables and $Y$ is a fully independent family of Bernoulli random variables.  The proof is essentially the same as in the Gaussian case with a few minor changes that need to be made.  In particular, the following steps do not carry over immediately:
\begin{enumerate}
\item The reduction from the case of a general polynomial to that of a multilinear polynomial
\item Theorem \ref{anticoncentrationTheorem} does not hold for Bernoulli random variables
\item Theorem \ref{LatalaThm} is not stated for the Bernoulli case
\end{enumerate}

The first of these problems is even easier to deal with in the Bernoulli case than in the Gaussian case.  This is because any degree-$d$ polynomial is equal to some degree-$d$ multilinear polynomial on the hypercube.

The second of these problems can be dealt with by fairly standard means.  In particular, the Invariance Principle of \cite{MOO} implies that for sufficiently regular polynomials, $p$, that $p(X)$ is anticoncentrated even for $X$ a Bernoulli random variable.  We are still left with the problem of reducing ourselves to the case of a regular polynomial.  This would be done using a regularity Lemma similar to that proven in \cite{regularity}, showing that an arbitrary polynomial threshold function can be written as a decision tree on a small number of coordinates such that most of the leaves are approximated by regular polynomial threshold functions.  Given a slight modification of this result telling us that these ``approximations'' hold even on $k$-independent inputs would allow us to reduce to the case of a regular polynomial after determining the values of $O_d(\epsilon^{-O(d)})$ coordinates.

The last of these concerns is apparently more significant, but can be dealt with by proving that Theorem \ref{LatalaThm} does hold for polynomials of Bernoullis.  In particular, one can show that a higher moment of a polynomial with respect to the Bernoulli distribution can be bounded in terms of the corresponding moment with respect to the Gaussian distribution.  In particular, we show that:
\begin{lem}
Let $p$ be a homogeneous degree-$d$ multilinear polynomial and $k\geq 1$.  Let $X$ be a Bernoulli random variable and $Y$ a Gaussian random variable.  Then
$$
\E[|p(X)|^k] = O(1)^{dk}\E[|p(Y)|^k].
$$
\end{lem}
\begin{proof}[Proof (Thanks to Jelani Nelson)]
Let $\sigma=(\sigma_1,\ldots,\sigma_n)$ be an $n$-dimensional Bernoulli random variable and $G=(g_1,\ldots,g_n)$ an $n$-dimensional Gaussian random variable independent of of $\sigma$.  Note that $\sigma_i |g_i|$ is distributed as a Gaussian.  Therefore we have that
$$
\E[|p(G)|^k] = \E[|p(\sigma_1|g_1|,\ldots,\sigma_n|g_n|)|^k]  = \E_G[ \E_\sigma[|p(\sigma_1|g_1|,\ldots,\sigma_n|g_n|)|^k]].
$$
By the convexity of the $L^k$ norm this is at least
$$
\E_\sigma\left[\left|\E_G[p(\sigma_1|g_1|,\ldots,\sigma_n|g_n|)]\right|^k\right].
$$
On the other hand, we have that
$$
\E_G[p(\sigma_1|g_1|,\ldots,\sigma_n|g_n|)] = \sqrt{\frac{2}{\pi}}^d p(\sigma).
$$
Therefore we have that
$$
\E[|p(G)|^k] \geq \sqrt{\frac{2}{\pi}}^{dk} \E_\sigma[|p(\sigma)|^k].
$$
As desired.
\end{proof}

\section{Conclusion}\label{kIndepConc}

The bounds on $k$ presented in this paper are far from tight.  At the very least the argument in Lemma \ref{expectationErrorLem} could be strengthened by considering a larger range of cases of $|p(x)|$ rather than just whether or not it is larger than $\sqrt{M}$.  At very least, this would give us bounds on $k$ of the form $O_d(\epsilon^{-x^d})$ for some $x$ less than 7.  I suspect that the correct value of $k$ is actually $O(d^2 \epsilon^{-2})$, and in fact such large $k$ will actually be required for $p(x)=\prod_{i=1}^d(\sum_{j=1}^{k} x_{i,j})$.  On the other hand, this bound is at the moment somewhat beyond our means.  It would be nice at least to see if a bound of the form $k=O_d(\epsilon^{-\textrm{poly}(d)})$ can be proven.  The main contribution of this work is prove that there is some sufficient $k$ that depends on only $d$ and $\epsilon$.

\section*{Acknowledgment}

This work was done with the support of an NSF graduate fellowship.

\end{document}